\documentclass[letterpaper, 10 pt, conference]{ieeeconf}

\IEEEoverridecommandlockouts                              

\usepackage[utf8]{inputenc}
\usepackage{include_packages} 
\usepackage[font=footnotesize,skip=5pt]{caption}
\usepackage[english]{babel}
\usepackage{xcolor}

\usepackage{todonotes}
\title{Game-to-Real Gap: Quantifying the Effect of Model Misspecification in Network Games}
\author{{Bryce L. Ferguson, Chinmay Maheshwari, Manxi Wu, and Shankar Sastry}
\thanks{B. L. Ferguson (bryce.l.ferguson@dartmouth.edu) is with the Thayer School of Engineering at Dartmouth College.}
\thanks{C. Maheshwari (chinmay\_maheshwari@jhu.edu) is with the Department of Electrical and Computer Engineering at Johns Hopkins University.}
\thanks{M. Wu (manxiwu@berkeley.edu) is with the Department of Civil and Environmental Engineering, and S. Sastry(shankar\_sastry@berkeley.edu) is with the Department of Electrical Engineering and Computer Sciences at the University of California, Berkeley.}%
}

\begin{document}
\maketitle 

\begin{abstract}
    Game-theoretic models and solution concepts provide rigorous tools for predicting collective behavior in multi-agent systems. In practice, however, different agents may rely on different game-theoretic models to design their strategies. As a result, when these heterogeneous models interact, the realized outcome can deviate substantially from the outcome each agent expects based on its own local model.
    In this work, we introduce the game-to-real gap, a new metric that quantifies the impact of such model misspecification in multi-agent environments. The game-to-real gap is defined as the difference between the utility an agent actually obtains in the multi-agent environment (where other agents may have misspecified models) and the utility it expects under its own game model.
    Focusing on quadratic network games, we show that misspecifications in either (i) the external shock or (ii) the player interaction network can lead to arbitrarily large game-to-real gaps. We further develop novel network centrality measures that allow exact evaluation of this gap in quadratic network games. Our analysis reveals that standard network centrality measures fail to capture the effects of model misspecification, underscoring the need for new structural metrics that account for this limitation. Finally, through illustrative numerical experiments, we show that existing centrality measures in network games may provide a counterintuitive understanding of the impact of model misspecification.
\end{abstract}




\section{Introduction}
The efficacy of {decision-making and control} algorithms within multi-agent settings is conditioned on the intentions, preferences, and capabilities of each individual agent.
In settings like autonomous driving, vehicles must predict the trajectories of other vehicles or pedestrians to ensure safe maneuvers~\cite{paccagnanNashWardropEquilibria2019, seff2023motionlm}.
When making investment and purchasing decisions, firms must forecast demand and the investment of other firms~\cite{candoganOptimalPricingPresence2010,acemogluNetworksShocksSystemic2015}.
When the number of agents is greater than two, {predictions} about individual agents are not sufficient~\cite{seff2023motionlm,albrecht2018autonomous}, instead predictions of collective behavior are needed.
Additionally, each agent may perform their own prediction; if these predictions are misaligned, it will degrade performance. 

Game theory regularly surfaces as a tool across many disciplines to model and reason about multi-agent interactions~\cite{schotter1980economics,Marden2014}.
Two common approaches to understanding the behavior of multi-agent systems are game-theoretic learning--in which players iteratively update their actions or policies in response to one another~\cite{fudenbergTheoryLearningGames1998}--and game-theoretic planning--in which an agent computes an equilibrium based on its conjectured game model, which is then used to compute its strategy~\cite{lavalle2000robot}.
In this work, we focus on the latter framework to capture decisions made by engineers and autonomous agents with significant lead time but with little opportunity to revise after deployment, e.g., an autonomous race-car which must develop a defending and overtaking policy in advance while preparing for race day event~\cite{kalaria2024alpha}, or distributed generator and storage facilities participating in smart grid demand management programs \cite{atzeniDemandSideManagementDistributed2013}.
Traditional game theory depends on fine assumptions about other players and can be sensitive to changes.
It is our intention to provide guarantees on the robustness of game-theoretical solutions to enable engineers to more confidently make the leap from theory to use in reality. 
To this end, we seek to develop formal analysis methodologies that will aid in promoting design techniques within multi-agent systems that are robust to mischaracterizations of other agents' intentions or capabilities.

We formalize the idea of game-theoretic planning by assigning each agent a predictive model (consisting of a game and a solution concept) with which they may leverage optimization techniques to devise their control policy.
We are interested in the case where the predictions provided by these models are inaccurate and heterogeneous among players due to mischaracterizations of one another's preferences and capabilities.
To quantify how sensitive an agent's predicted performance is to these types of misspecifications, we introduce the \emph{\gapname}, defined as the difference between an agent's predicted and realized cost.
Our goal is to characterize the gap and identify features around which uncertainty or misspecification may contribute greater loss in performance.
As a first step towards this goal, we focus on the class of network games and the Nash equilibrium solution concept.

Despite its relevance to real-world planning, to the authors' best knowledge, the setup that each player devises their action as the solution to a conjectured game independently and heterogeneously from one another has been sparsely studied.
\cite{kreps1990game} suggests that human players follow conjectured game models in their decision-making, and \cite{DEVETAG2008364} affirms this and finds that human players tend to act optimally within simplified mental game models.
The authors of \cite{meir2015playing} adopt a similar perspective as this work of different players having un-modeled biases or preferences in congestion games, and in \cite{zwillinger2023gametheoryanalysisplaying}, the authors study a similar setup but in two-player matrix games.
These prior findings motivated this work to rigorously quantify the consequences of conjectured games.
That notwithstanding, the concept of sensitivity, robustness, and uncertainty has been studied broadly in game theory and control, though with different motivations and different contexts.
Our work is related to the recent line of literature on sensitivity analysis of Nash equilibrium in network games \cite{parise2019variational, parise2017sensitivity}. This literature quantifies the change in equilibrium with respect to the change in the underlying parameters of the game. Meanwhile, we quantify the gap in experienced cost by players when each player has a different estimate of the underlying game parameters.
Some recent works propose an adaptive scheme where players infer the utility function of others dynamically~\cite{soltanian2025pace,ward2025active}.
Other frameworks seek to embed uncertainty within the game itself; in Bayesian games, introduced in \cite{harsanyi1967games}, players reason about uncertainty in the game parameters.
Another related line of literature is that on robust game theory \cite{aghassi2006robust} where each player computes an equilibrium strategy by assuming worst-case realization of uncertain parameters.
In any game model, the usefulness of the solution concept is subject to the accuracy of the primitives, e.g., priors in Bayesian games or {worst-case uncertainty realization} in robust games.
{Our work seeks to characterize the change in cost experienced by players due to misspecification in game models (or primitives) between players.}


{The main contributions of this work are summarized as follows. In \cref{subsec:conjectures}, we formally introduce the notion of the game-to-real gap (see Definition \ref{def: GameToRealGap}). In Section \ref{sec: EvaluatingGTRNetwork}, we analyze the game-to-real gap within the context of quadratic network games. Such games are characterized by utilities comprising three components: (i) a quadratic term representing the magnitude of each player’s local strategy; (ii) an interaction term given by the inner product between a player’s strategy and the aggregate strategy of others, weighted by an interaction matrix; and (iii) the inner product of the player’s strategy with a player-specific external shock. We show that the game-to-real gap can be made arbitrarily large—even under arbitrarily small misspecifications in the external shock parameters (Proposition \ref{prop:unbounded_shock}) or the interaction matrix (Proposition \ref{prop:unbounded_graph}).

Next, we propose two novel network centrality measures—Shock Misspecification Centrality (Definition \ref{def:overlap}) and Interaction-Graph Misspecification Centrality (Definition \ref{def: InteractionGraphCentrality})—which allow us to precisely characterize the game-to-real gap under shock misspecification (Theorem \ref{thm:network_agg_bound_shock}) and interaction-matrix misspecification (Proposition \ref{prop:network_agg_bound_net}), respectively. Finally, in \cref{subsec:num}, we present numerical experiments that illustrate these theoretical findings. Notably, we demonstrate that misalignment with a player of high Bonacich centrality (a common centrality measure in network games) does not necessarily correlate with significant degradation in predicted performance, thereby motivating the need for the proposed centrality measures.}
\section{Modeling Game-to-Real Gap}
Let $P_{i,j}$ denote the $(i,j)$ block of $P$ and $P(l,k)$ denote the $(l,k)$ element.
Let $-i := N\setminus\{i\}$.
When referencing an object conjectured by a player $j \in N$, we will use the notation $\persp{(\cdot)}{j}$.

\subsection{Games and Solution Concepts}
In this work, we adopt the perspective that a game is a model for multi-agent interaction, that, when paired with a solution concept, offers an analytical approach to predicting the behavior of a group of decision-makers.
Formally, let $G = (N,\mathcal{U},\{J_i\}_{i \in N})$ denote a game where $N = \{1,\ldots,n\}$ is a finite set of players, $u_i \in \mathcal{U}_i$ is the control action of player $i\in N$ and $\mathcal{U} = \mathcal{U}_1\times\ldots\times\mathcal{U}_n$ is the joint action space, and $J_i:\mathcal{U}\rightarrow \mathbb{R}$ is the cost function of player $i \in N$.
On its own, a game captures the primitives and preferences of interactions.
The emergent behavior we expect or may seek through the design of local decision-making algorithms is typically captured by some notion of equilibrium or solution concept of the form $\soln(G) \subseteq \mathcal{U}$ which provides a subset of the joint-action space deemed to be strategically plausible.
Many examples of solution concepts exist in game theory research, including Nash equilibria, correlated equilibria, strong equilibria, and more~\cite{ aumann1959acceptable}, each of which captures different notions of strategic plausibility.
Together, a game model and solution concept pair form a predictive model of collective behavior we can expect from a multi-agent interaction.
{In this work, we aim to gain a deeper understanding of the robustness of decisions made using these models under heterogeneous model misspecifications among the players.}


\begin{figure}
    \centering
    \includegraphics[width=0.99\linewidth]{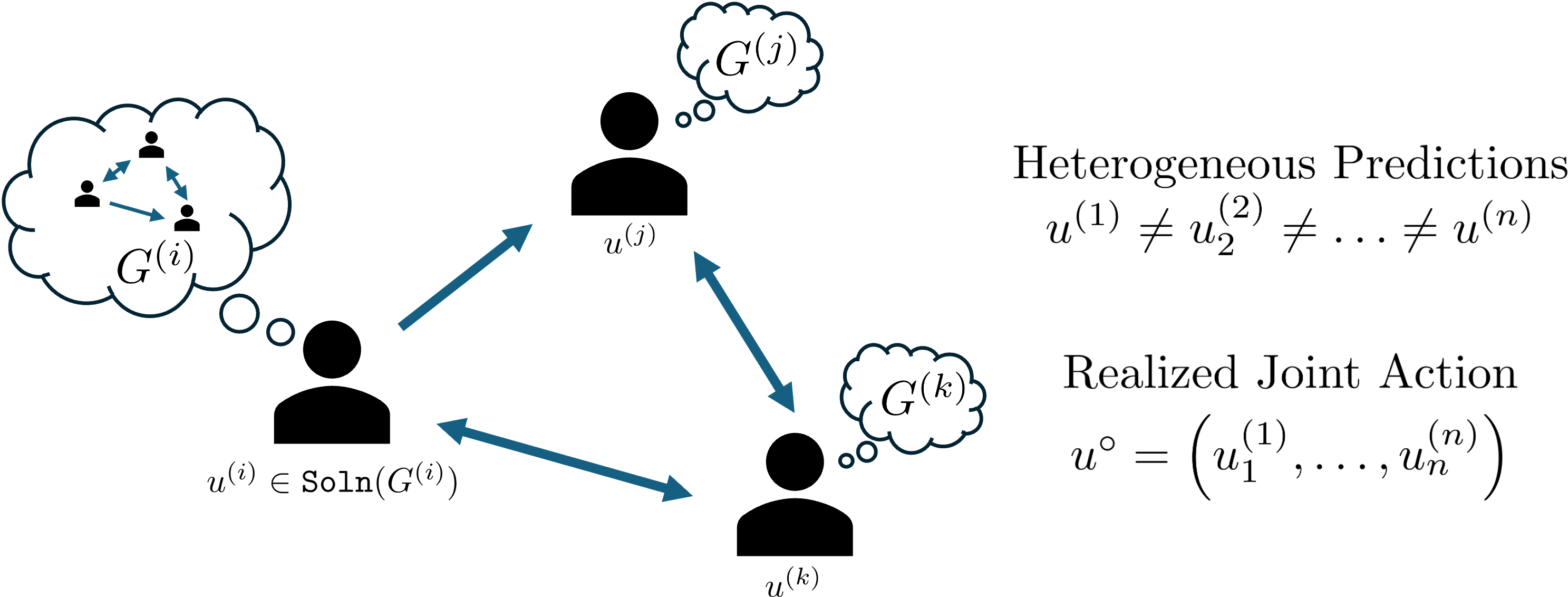}
    \caption{Illustration of misaligned game theoretic planning. Each player $i \in N$ conjectures a game $\persp{G}{i}$ and devises their action from a solution concept $\soln$, e.g., Nash equilibrium. If players' conjectures are misaligned, the realized joint action need not be an equilibrium and can offer degraded performance relative to each player's prediction.}
    \label{fig:conjecture_illustration}
\end{figure}

\subsection{Conjectured games and the \gapname}\label{subsec:conjectures}
When a decision maker wishes to plan their control action using a game-theoretic model, they must conjecture the game that they expect to face in the real world.
\begin{definition}
    The \emph{conjectured game} by player \(i\in N\) is denoted by the tuple $\persp{G}{i} = (N,\persp{\mathcal{U}}{i},\{\persp{J_j}{i}\}_{j \in N})$, where $\persp{\mathcal{U}}{i} = \prod_{j\in N}\persp{\mathcal{U}}{i}_j$ is the joint action set conjectured by player \(i\) and $\persp{J_j}{i}: \persp{\mathcal{U}}{i} \rightarrow \mathbb{R}$ denotes the cost of player \(j\in N\) conjectured by player \(i\).
    Each player knows their own cost and thus $\persp{J_i}{i} = J_i$.
\end{definition}
We formalize the agents' decision via a game-theoretic solution, i.e., we will say that player $j$ predicts the joint action to be $\persp{u}{j} \in \soln(\persp{G}{j})$ and takes the action {$\persp{u_j}{j}$}.

Together, a solution concept and game model pair $(\texttt{Soln},G)$ form a predictive model that an agent may use in their decision-making and planning.
Still, planning is not done by one single agent; instead, each agent must predict and make decisions independently, resulting in a set of predictive models $\{(\texttt{Soln}^{(j)},\persp{G}{j})\}_{j \in N}$ as illustrated in~\cref{fig:conjecture_illustration}.
When each player institutes their control action with respect to their own model, they deploy their personal component of their predicted joint action, i.e. $\persp{u_i}{i}$.
If each player has no misspecification, then each agent will receive an anticipated cost for their action.
However, suppose even one player has a small misspecification of the cost function or capabilities of another. In that case, it can misalign their prediction and negatively impact the realized cost of itself and other players, even though others have no such misspecification.
{Misspecifications can emerge from modeling errors, noise, perturbations, or assumptions}. We are thus motivated to study this loss from predicted to realized performance when misspecifications cause misalignments in models of strategic interactions.
We formalize this discrepancy via the following:
\begin{definition}\label{def: GameToRealGap}
    For a set of conjectured games $\{ \persp{G}{i}\}_{i \in N}$, the \gapname~ for player $i \in N$ is defined to be 
$$ {J_i(u^\circ) - J_i(\persp{u}{i}),}$$
where $u^\circ := (\persp{u_1}{1},\ldots,\persp{u_i}{i},\ldots,\persp{u_n}{n})$,
which is the difference between $i$'s realized cost and what they predicted via their game model.
\end{definition}
Our goal is to characterize the increase in cost caused by model misspecification.
We aim to gain initial insights from this perspective in the broadly studied class of network games.


\section{Evaluating Game-to-Real Gap in Network Games}\label{sec: EvaluatingGTRNetwork}
\subsection{Setup: Network Games}\label{ssec: NetworkGamesSetup}
Network games model settings where agents respond to some environmental stimulus while participating in multiple pairwise interactions with other ag,ents modeled by the edges of a directed network.
Positive and negative edge weights model whether one agent's actions are a strategic complement or substitute of another, i.e., if their objectives are positively or negatively correlated with one another.
A player's cost is the \mbox{sum of their internal cost and these interactions.}

We consider the multi-variate case where an agent $i \in N$ takes a vector action $u_i \in \mathbb{R}^{m_i}$.
The joint action is the concatenation of these decisions $u \in \mathbb{R}^m$ where $m = \sum_{i=1}^{n}m_i$.
Agent $i$ receives a reward for aligning their action with an external shock $\epsilon_i\in \mathbb{R}^{m_i}$ as well as through the pairwise interaction between agents.
From player $j$, player $i$ receives payoff $u_i^\top P_{i,j}u_j$, where $P_{i,j} \in \mathbb{R}^{m_i\times m_j}$.
One can consider this as a matrix weighted edge from agent $j$ to $i$ or that $P_{i,j}(l,k)$ is the weight of the edge from $u_j(k)$ to $u_i(l)$.

Agent $i$'s total cost is the quadratic cost of their local action, less the interactions with other agents and the environment, defined by
\begin{equation}
    J_i(u) = \frac{1}{2} u_i^\top u_i - u_i^\top\left({P_{i,-}} u + {\epsilon_i}\right)
\end{equation}
where ${P_{i,-}} = [P_{i,1},\ldots,P_{i,n}] \in \mathbb{R}^{m_i \times m}$ is the $i$th row of blocks of the \emph{interaction matrix} $${P} = \begin{bmatrix}
    P_{1,1} & \ldots & P_{1,n}\\\vdots & \ddots & \ddots\\P_{n,1} & \ldots &  P_{n,n}
\end{bmatrix} \in \mathbb{R}^{m\times m}.$$
A network game is thus captured by the tuple $G = (N, \{m_i\}_{i \in N},P,\epsilon)$.
As defined in this work, network games are instances of aggregate games with weighted, linear aggregators~\cite{belgioiosoConvexityMonotonicityGeneralized2017,parise2019variational}, and have been used to model a variety of multi-agent interactions, including markets and supply chains~\cite{candoganOptimalPricingPresence2010,acemogluNetworksShocksSystemic2015},
power grid storage and generation \cite{atzeniDemandSideManagementDistributed2013},
and transportation and congestion routing
\cite{paccagnanNashWardropEquilibria2019}.
In each of these settings are numerous sources of noise and parameters that may only be modeled via conjecture or estimation.
When agents make decisions using game models, they may possess inaccuracies in their primitives, leading to discrepancies in their predicted behavior.

In this work, we consider two forms of misspecifications: {environmental misspecifications}, captured by heterogeneous predictions of the parameter $\epsilon$ across the network, and agent-to-agent misspecifications, captured by the misalignment of the conjectured interaction matrix $P$.
The predominant solution concept in network games is the Nash equilibrium.
\begin{definition}
    A joint action $u \in \mathbb{R}^m$ is a Nash equilibrium in the network game $G$ if $J_i(u) \leq J_i(u_i^\prime,u_{-i})$ for all $u_i^\prime  \in \mathbb{R}^{m_i}$ and $i \in N$.
\end{definition}
When $(P+P^\top)/2 \prec I$, the game is strongly monotone and the Nash equilibrium is unique and can be expressed in closed form~\cite{acemogluNetworksShocksSystemic2015,rosenExistenceUniquenessEquilibrium1965}.

In line with \cref{subsec:conjectures}, we will denote the conjectured environmental network shock and conjectured interaction matrix from the perspective of player $j$ as $\persp{\epsilon}{j}$ and $\persp{P}{j}$, respectively.
For player $j$, who uses this solution concept alongside their conjectured game $\persp{G}{j}$, they predict the joint action
\begin{equation}\label{eq:net_agg_equilibrium}
    \persp{u}{j} = (I-\persp{P}{j})^{-1}\persp{\epsilon}{j},
\end{equation}
from which, they extract their action $\persp{u_j}{j}$.
\begin{assumption}
    We assume throughout that $(\persp{P}{j}+\persp{P}{j}\hspace{0pt}^\top)/2 \prec I$ is satisfied for all $j \in N$.
\end{assumption}


\subsection{Network games with arbitrarily large game-to-real gap}\label{subsec:cons}
To more rigorously understand the impact of misspecifications in strategic models, we first seek to identify under what conditions performance losses can occur and to what degree.
We seek to answer if, given some bounded level of mischaracterization among players, can we upper bound the \gapname?
The short answer is no.
In the following propositions, we demonstrate that the \gapname~ can be significant in the presence of any non-zero misalignment of the conjectured network games.
We focus on two types of mischaracterization, respectively: heterogeneous predicted shocks ($\persp{\epsilon}{i} \neq \persp{\epsilon}{j}$), and misaligned interaction matrix ($\persp{P}{i} \neq \persp{P}{j}$); in either case, we demonstrate that the increase in realized cost \mbox{from misspecifications can be large.}

In Proposition~\ref{prop:unbounded_shock}, we show that, even in the case where the difference between each pair of agents' predicted network shocks $\epsilon$ is upper bounded by a positive constant, the gap between predicted and realized performance can be arbitrarily large.
\begin{proposition}\label{prop:unbounded_shock}
    In quadratic network aggregative games, for any $\delta,M > 0$, there exists some $\{\persp{\epsilon}{i}\}_{i \in N}$ and interaction matrix $P$ such that
    $$\lVert \persp{\epsilon}{i} - \persp{\epsilon}{j}\rVert_2 \leq \delta \quad \forall \ i, j \in N,
    $$
    but that for each $i \in N$,
    $$J_i(u^\circ) - J_i(\persp{u}{i}) > M.$$
\end{proposition} 
This result is notable because network games are a well-studied class of games with desirable analysis properties, such as uniqueness of equilibria~\cite{rosenExistenceUniquenessEquilibrium1965} and continuity of the solution map over payoff-relevant parameters~\cite{parise2017sensitivity}.
Nevertheless, the accuracy of an agent's predicted cost can be extremely sensitive to misalignment between its predicted network shocks $\epsilon$ and those of other players (a variable in which the users' cost functions are linear).
The misalignment of these terms in the real world is also entirely plausible, in economic settings, if firms possess heterogeneous forecasting models of upcoming price volatility, they may adopt heterogeneous predictions.
Similarly, a collection of autonomous agents may possess local sensors collecting heterogeneous data used to form estimates about environmental states (e.g., predicted renewable energy generation).

A proof of Proposition \ref{prop:unbounded_shock} is in the Appendix. The proof follows by construction; however, it reveals that networks where $\sigma_{\rm min}(I-P)$ is small, e.g., a cycle with near unitary edge weights, are more susceptible to large \gapname.

Another aspect in modeling strategic interactions is not just predicting environmental payoff-relevant parameters, but also conjecturing how agents interact with one another and how they perceive one another's actions.
To capture how conjectured interactions may manifest in the context of network games, we consider that each player $i \in N$ conjectures an interaction network $\persp{P}{i}$ which they believe models the connections between all the players.
As before with conjectured network shocks, two players may have misaligned conjectures on the network of interactions; accordingly, each player's predicted joint-action may differ, resulting in a \gapname.
In Proposition~\ref{prop:unbounded_graph}, we show that with arbitrarily small differences in players' conjectured networks, there exists a possibility for arbitrarily large differences between predicted and realized cost.
\begin{proposition}\label{prop:unbounded_graph}
    In quadratic network aggregative games, for any $\delta, M > 0$,  there exists some $\epsilon$ and conjectured interaction matrices $\{\persp{P}{i}\}_{i \in N}$
    such that for all $i,j \in N$,
    $$\lVert \persp{P}{i} - \persp{P}{j}\rVert_F \leq \delta,$$
    but that for each player $i \in N$
    $$J_i(u^\circ) - J_i(\persp{u}{i}) > M.$$
\end{proposition}
A proof of Proposition \ref{prop:unbounded_graph} appears in the Appendix and also follows by construction.

Though we have identified that the \gapname~can be large under bounded misalignments, this does not mean that every problem instance admits similar susceptibility.
In \cref{subsec:characterization}, we seek to refine this understanding of which multi-agent systems admit sensitive performance predictions based on problem parameters and network properties.

\subsection{Exact characterization of the \gapname}\label{subsec:characterization}
In \cref{subsec:cons}, it was revealed that the consequences of misspecifications can be significant, even when said misspecifications are small.
In this section, we seek to refine this by identifying properties of network games which cause them to admit a large \gapname.
To do so, we introduce a new notion of node communicability, which extends existing notions of Katz-Bonacich centrality to measure how significant the overlap in network effect is between the two agents' action vectors.
To do so, we define the Leontief matrix $L := (I-P)^{-1}$, which in networks is used to characterize a node's centrality~\cite{Katz_1953, bonacichPowerCentralityFamily1987}.
Denote the Bonacich centrality of a single control action $u_i(k)$ by $\beta_i(k) = \sum_{l = 1}^nL_{i,-}(k,l)$ or the row sum of the $u_i(k)$ row of the Leontief matrix.

In network games, the Bonacich centrality of a node (syn. player) corresponds to the weighted average of paths originating at node $i$ to each other node in the network, which quantifies how a shock originating at $i$ may propagate throughout the network.
It has been shown that this and other notions of centrality correlate with the sensitivity of the Nash equilibrium to the shock~\cite{acemogluNetworksShocksSystemic2015,parise2017sensitivity}.
However, in seeking to understand how the misalignment of multiple agents' conjectures affects predicted performance, we must consider the inter-agent effects through the network.

\subsubsection{Heterogeneous Shock Forecasts}
We first seek to more precisely characterize the \gapname~in the case where players have different predictions of the network shock $\epsilon$.
To do so, we introduce a new operator on the adjacency matrix $P$ offering a pairwise perspective on centrality.
\begin{definition}\label{def:overlap}
The \emph{Shock Misspecification Centrality} between player $i$ and $j$ is denoted by
\begin{align*}
    \mathcal{B}_{i,j} &= L_{i,-}^\top P_{i,j} L_{j,-}\\
    &= [(I-P)^{-1}]_{i,-}^\top P_{i,j} [(I-P)^{-1}]_{j,-} \in \mathbb{R}^{m\times m}.
\end{align*}
\end{definition}
Intuitively, $\mathcal{B}_{i,j}$ captures the edge-conditioned overlap of two players' centrality profile $L_{i,-}$ and $L_{j,-}$.
This overlap profile allows us to naturally generalize to agents with vector actions, even of different sizes.
When $P$ is symmetric, then $\mathcal{B}_{i,j}$ is additionally the directional derivative of the Leontif matrix $L$ in the direction $P_{i,j}$.
When each player's action is a scalar, then $\mathcal{B}_{i,j}$ is proportional to the outer product of their centrality profiles $\beta_i$ and $\beta_j$.

In \cref{thm:network_agg_bound_shock}, we characterize the \gapname~in closed form via the Shock Misspecification Centrality and predicted shocks between one player and each other; further, we upper bound the gap when the distance between any two predictions is bounded by a constant factor and observe that this bound is conditioned on the inverse of $\sigma_{\rm min}(I-P)$, matching the intuition from \cref{subsec:cons} that the \gapname~can be large when this quantity is small.
We further note that the gain in cost need not occur from two players misconjecturing about one another; it can also be contributed to by two players having misaligned predictions of a third.
\begin{theorem}\label{thm:network_agg_bound_shock}
In quadratic network aggregative games with network $P$ such that $(P+P^\top)/2 \prec I$, if each player has a heterogeneous conjecture of the shock $\{\persp{\epsilon}{j}\}_{j \in N}$ 
Then the \gapname~for player $i$ is
    \begin{align}\label{eq:thm_bound}
        J_i(u^\circ) - J_i(\persp{u}{i}) 
        &= \sum_{j \neq i} \persp{\epsilon}{i}\hspace{0pt}^\top \mathcal{B}_{i,j}(\persp{\epsilon}{i}-\persp{\epsilon}{j}).
    \end{align}
Further, if for all $i,j \in N$, 
$\lVert \persp{\epsilon}{i} - \persp{\epsilon}{j}\rVert_2 \leq \delta,$
then
\begin{align*}
    J_i(u^\circ) - J_i(\persp{u}{i}) \hspace{-1pt} \leq \hspace{-1pt}\Big(\delta\lVert \persp{\epsilon}{i}\rVert_2 \sum_{j \neq i}\lVert P_{i,j}\rVert_2\Big)\hspace{-2pt}/\hspace{-2pt}\left(\sigma_{\rm min}(I-P)\right)^2.
\end{align*}
\end{theorem}
\begin{proof}
    Consider the case where each player $i \in N$ conjectures a game $\persp{G}{i}$ which has adjacency matrix $P$ and predicted disturbances $\persp{\epsilon}{i} \in \mathbb{R}^m$.
    To capture the profile of all players' predictions, let $\overline{\epsilon} \in \mathbb{R}^{nm}$ be the vector concatenating each player's respective predicted disturbances.

    We define the realized joint action $u^\circ(\overline{\epsilon})$ as a function over the set of prediction profiles to joint actions.
    Note that when $\overline{\epsilon} = [\persp{\epsilon}{i}\hspace{0pt}^\top,\ldots,\persp{\epsilon}{i}\hspace{0pt}^\top]^\top$, then 
    $u^\circ(\overline{\epsilon}) = \persp{u}{i}$ due to the uniquenss of equilibria and homogeneity of the conjectured games; we define this specific prediction profile as $\overline{\persp{\epsilon}{i}}$ and the realized joint actions from each players' predictions as $\overline{\epsilon^\circ} := [\persp{\epsilon}{1}\hspace{0pt}^\top,\ldots,\persp{\epsilon}{n}\hspace{0pt}^\top]\hspace{0pt}^\top$.
    To characterize the gap in player $i$'s predicted and realized cost, we will take the line integral of $J_i(u^\circ(\overline{\epsilon}))$ between $\overline{\persp{\epsilon}{i}}$ and $\overline{\epsilon^\circ}$.
    We can express this line integral by
    \begin{subequations}
    \begin{align}
        &J_i(u^\circ) - J_i(\persp{u}{i})\label{eq:line_int_a} \\
        &= \int_{s=0}^1 \left(\overline{\epsilon^\circ} - \overline{\persp{\epsilon}{i}}\right)^\top \nabla_{\overline{\epsilon}} J_i\left(u^\circ\left((1\textup{-}s)\overline{\persp{\epsilon}{i}} + s\overline{\epsilon^\circ}\right)\right) ds \label{eq:line_int_b}\\
        &= \int_{s=0}^1 \left(\overline{\epsilon^\circ} - \overline{\persp{\epsilon}{i}}\right)^\top \nonumber \\
        &\Bigg(\hspace{-4pt}\begin{bmatrix}
            \nabla_\epsilon \persp{u_1}{1}\left((1\textup{-}s)\persp{\epsilon}{i} \textup{+} s\persp{\epsilon}{1}\right)^\top \nabla_{u_1} J_i\left(u^\circ\left((1\textup{-}s)\overline{\persp{\epsilon}{i}} \textup{+} s\overline{\epsilon^\circ}\right)\hspace{-1pt}\right)\\
            \vdots\\
            \nabla_\epsilon \persp{u_n}{n}\left((1\textup{-}s)\persp{\epsilon}{i} \textup{+} s\persp{\epsilon}{n}\right)^\top \nabla_{u_n} J_i\left(u^\circ\left((1\textup{-}s)\overline{\persp{\epsilon}{i}} \textup{+} s\overline{\epsilon^\circ}\right)\hspace{-1pt}\right)
        \end{bmatrix} \nonumber\\
        &+ \nabla_{\overline{\epsilon}}J_i\left(u^\circ;\left((1\textup{-}s)\overline{\persp{\epsilon}{i}} + s\overline{\epsilon^\circ}\right)\right) \Bigg) ds \label{eq:line_int_c}\\
        &= \int_{s=0}^1 \left( \overline{\persp{\epsilon}{i}} - \overline{\epsilon^\circ}\right)^\top \begin{bmatrix}
            [(I-P)^{-1}]_{1,-}^\top P_{i,1}^\top\\
            \vdots\\
            [(I-P)^{-1}]_{n,-}^\top P_{i,n}^\top
        \end{bmatrix} \persp{u_i}{i} ds \label{eq:line_int_d}\\
        &= \sum_{j \neq i} (\persp{\epsilon}{i} - \persp{\epsilon}{j})^\top [(I-P)^{-1}]_{j,-}^\top P_{i,j}^\top \persp{u}{i}_i \label{eq:line_int_e}.
    \end{align}
    \end{subequations}
    Where \eqref{eq:line_int_d} holds by noting that $J_i$ is parameterized by $\overline{\epsilon}$, but $\nabla_{\persp{\epsilon}{j}_k} J_i(u;\overline{\epsilon}) = \mathbf{0}_{m_k}$ if $j \neq i$ and $\left(\overline{\persp{\epsilon}{i}} - \overline{\epsilon^\circ}\right)_k = 0$ for all $k$ in the $i$th block.
    Accordingly, the latter term in \eqref{eq:line_int_c} evaluates to zero.
    Additionally, the chain rule can be evaluated explicitly by recalling that $\nabla_\epsilon \persp{u}{j}_j = [(I-P)^{-1}]_{j,-}$ for all $\epsilon$ and $\nabla_{u_j} J_i(u) = -P_{i,j}^\top u_i$ for all $j \neq i$.
    \eqref{eq:line_int_e} holds from evaluating the integral over an integrand that no longer depends on the integrating variable and substituting the definition of $\persp{u_i}{i}$.
    Transposing provides \eqref{eq:thm_bound}

    To complete the proof, simply apply Cauchy-Schwarz inequality and the property of the induced 2-norm of the matrix operator $I-P$ being greater than the induced 2-norm of any block and upper bound $\lVert \persp{\epsilon}{j} - \persp{\epsilon}{i}\rVert_2$ by $\delta$.
\end{proof}

\subsubsection{Heterogeneous Network Models}
Proposition~\ref{prop:unbounded_graph} demonstrated that players' misaligned conjectures about how they interact with one another, i.e., when $\persp{P}{i} \neq \persp{P}{j}$, can also introduce performance losses.
We introduce a new overlap profile that considers misaligned networks.
\begin{definition}\label{def: InteractionGraphCentrality}
    The Interaction-Graph Misspecification Centrality between player $i$ and player $j$ is denoted by
\end{definition}
\begin{align*}
    \mathcal{C}_{i,j} = \persp{L}{i}_{i,-}\hspace{0pt}^\top(\persp{P}{j}-\persp{P}{i})_{i,j}(\persp{L}{i}\persp{P}{i}\persp{L}{j})_{j,-}\\
\end{align*}
This new operator on the pair $\persp{P}{i}$ and $\persp{P}{j}$ resembles the introduced Shock Misspecification Centrality but differs in two key ways: the second term is now the difference between the two considered players' conjectured adjacency matrices, and the third term is now the $\persp{P}{i}$ weighted product of each player's conjectured Leonteif matrix.
Each of these changes captures that difference between the two graphs.
When we consider the case where $\persp{P}{i} = \persp{\alpha}{i}P$, then at the point when $\persp{\alpha}{i} = \persp{\alpha}{j}$, $\persp{L}{i}\persp{P}{i}\persp{L}{i}$ is $\partial \persp{L}{i} / \partial \alpha$.

In Proposition~\ref{prop:network_agg_bound_net}, we characterize the \gapname~caused by heterogeneous conjectured networks using the introduced Interaction Graph Misspecification Centrality.
Similar intuition is revealed that the loss in predicted performance does not happen solely from alignment with central nodes, but from alignment with similar nodes in the network.
We exemplify this in \cref{subsec:num}.
Further bounding the gap by its's primitives reveals a similar dependence on the recipricol of $\sigma_{\rm min}(I-P)$, though now considering each players conjectured network.
\begin{proposition}\label{prop:network_agg_bound_net}
In quadratic network aggregative games where each player $i \in N$ conjectures a network $\persp{P}{i} = \persp{\alpha}{i}P$ with $\persp{\alpha}{i} > 0$ and $\persp{\alpha}{i}(P+P^\top)/2 \prec I$, if the shock to the network is $\epsilon \in \mathbb{R}^m$, then the \gapname~for player $i$ is 
\begin{equation}\label{eq:prop_graph_char}
J_i(u^\circ) - J_i(\persp{u}{i}) = \sum_{j \neq i} \epsilon^\top \mathcal{C}_{i,j}\epsilon.
\end{equation}
Further, if for each $i,j \in N$, $|\persp{\alpha}{i}-\persp{\alpha}{j}|\leq \delta$, then
\begin{equation*}
    J_i(u^\circ) - J_i(\persp{u}{i}) \leq \sum_{j \neq i}\frac{\delta |\persp{\alpha}{i}| \| P\|_2\|P_{i,j}\|_2
    }{\sigma_{\rm min}(I-\persp{\alpha}{j}P)\sigma_{\rm min}(I-\persp{\alpha}{i}P)^2}
\end{equation*}
\end{proposition}
\begin{proof}
The proof follows similarly to the proof of \cref{thm:network_agg_bound_shock}, but with modifications for the heterogeneous network models.
    Consider the case where each player $i \in N$ conjectures a game $\persp{G}{i}$ which has adjacency matrix $\persp{\alpha}{i}P$ and predicted disturbances $\epsilon \in \mathbb{R}^m$.
    To capture the profile of all players' conjectured networks, let $\overline{\alpha} \in \mathbb{R}^{n}$ be the vector concatenating each player's respective predicted disturbances.

    We define the realized joint action $u^\circ(\overline{\alpha})$ as a function over the set of prediction profiles to joint actions.
    Note that when $\overline{\alpha} = [\persp{\alpha}{i}\hspace{0pt}^\top,\ldots,\persp{\alpha}{i}\hspace{0pt}^\top]^\top$, then 
    $u^\circ(\overline{\alpha}) = \persp{u}{i}$ due to the uniquenss of equilibria and homogeneity of the conjectured games; we define this specific prediction profile as $\overline{\persp{\alpha}{i}}$ and the realized joint actions from each players' predictions as $\overline{\alpha^\circ} := [\persp{\alpha}{1}\hspace{0pt}^\top,\ldots,\persp{\alpha}{n}\hspace{0pt}^\top]\hspace{0pt}^\top$.
    To characterize the gap in player $i$'s predicted and realized cost, we will take the line integral of $J_i(u^\circ(\overline{\alpha}))$ between $\overline{\persp{\alpha}{i}}$ and $\overline{\alpha^\circ}$.
    We can express this line integral by
\begin{subequations}
\begin{align}
&J_i(u^\circ) - J_i(\persp{u}{i})\label{eq:line_int_graph_a} \\
    &= \int_{s=0}^1 \left(\overline{\alpha^\circ} \textup{-} \overline{\persp{\alpha}{i}}\right)^\top \nabla_{\overline{\alpha}} J_i\left(u^\circ\left((1\textup{-}s)\overline{\persp{\alpha}{i}} \textup{+} s\overline{\alpha^\circ}\right)\right) ds \label{eq:line_int_graph_b}\\
    &= \int_{s=0}^1 \left(\overline{\alpha^\circ} \textup{-} \overline{\persp{\alpha}{i}}\right)^\top \nonumber \\
    &\hspace{-4pt}\Bigg(\hspace{-5pt}\begin{bmatrix}
        \frac{\partial}{\partial\alpha} \persp{u_1}{1}\left((1\textup{-}s)\persp{\alpha}{i} \textup{+} s\persp{\alpha}{1}\right)^\top\hspace{-3pt} \nabla_{u_1} \hspace{-2pt}J_i\left(u^\circ\left((1\textup{-}s)\overline{\persp{\alpha}{i}} \textup{+} s\overline{\alpha^\circ}\right)\hspace{-1pt}\right)\\
        \vdots\\
        \frac{\partial}{\partial\alpha} \persp{u_n}{n}\left((1\textup{-}s)\persp{\alpha}{i} \textup{+} s\persp{\alpha}{n}\right)^\top \hspace{-3pt} \nabla_{u_n}\hspace{-2pt} J_i\left(u^\circ\left((1\textup{-}s)\overline{\persp{\alpha}{i}} \textup{+} s\overline{\alpha^\circ}\right)\hspace{-1pt}\right)
    \end{bmatrix} \nonumber\\
    &+ \nabla_{\overline{\alpha}}J_i\left(u^\circ\left((1\textup{-}s)\overline{\persp{\alpha}{i}} \textup{+} s\overline{\alpha^\circ}\right);(1\textup{-}s)\overline{\persp{\alpha}{i}} \textup{+} s\overline{\alpha^\circ}\right) \Bigg) ds \label{eq:line_int_graph_c}\\
    &= \sum_{j\neq i} (\persp{\alpha}{i}\textup{-}\persp{\alpha}{j}) \epsilon^\top\label{eq:line_int_graph_d}\\
    &\cdot\left[\int_{s=0}^1 \frac{\partial}{\partial \alpha} \left(I- ((1\textup{-}s)\overline{\persp{\alpha}{i}} \textup{+} s\overline{\alpha^\circ})P \right)^{-1} ds\right]_{j,-}^\top \persp{\alpha}{i}P_{i,j}^\top \persp{u}{i}_i\nonumber\\
    &\hspace{-4pt}= \textup{-}\sum_{j\neq i} \persp{\alpha}{i}\epsilon^\top\left[(I-\persp{\alpha}{j}P)^{-1} - (I-\persp{\alpha}{i}P)^{-1}\right]_{j,-}^\top P_{i,j}^\top \persp{u}{i}_i \label{eq:line_int_graph_e}\\
    &=\sum_{j\neq i} \persp{\alpha}{i}(\persp{\alpha}{j} \textup{-} \persp{\alpha}{i})\epsilon^\top\nonumber\\
    &\cdot\left[(I-\persp{\alpha}{j}P)^{-1} P (I-\persp{\alpha}{i}P)^{-1}\right]_{j,-}^\top P_{i,j}^\top \persp{u}{i}_i \label{eq:line_int_graph_f},
\end{align}
\end{subequations}
where \eqref{eq:line_int_graph_b} and \eqref{eq:line_int_graph_c} hold from integrating the directional derivative of $J_i$ over the line segment in the $\overline{\alpha}$ space between $\overline{\persp{\alpha}{i}}$ and $\overline{\alpha^\circ}$. 
\eqref{eq:line_int_graph_d} holds from $\nabla_{\overline{\alpha}(j)}J_i(u;\overline{\alpha}) = 0$ for all $j \neq i$ and the entry where $j = i$ is multiplied by $\overline{\alpha^\circ} - \overline{\persp{\alpha}{i}} = 0$ and from $\nabla_{u_j}J(u;\alpha) = -\persp{\alpha}{i}P_{i,j}^\top\epsilon$.
\eqref{eq:line_int_graph_e}
\eqref{eq:line_int_graph_f}
by using the Resolvent Identity,
$(I-P)^{-1} - (I-B)^{-1} = (I-P)^{-1}(B-P)(I-B)^{-1}$.
We recover \eqref{eq:prop_graph_char} by transposing, substituting $\persp{u_i}{i} = L_{i,-}^\top\epsilon$, and grouping each $P_{i,j}$ with $\persp{\alpha}{j}-\persp{\alpha}{i}$ to get $\persp{P}{j}-\persp{P}{i}$ and $P$ with $\persp{\alpha}{i}$ to get $\persp{P}{i}$.

To complete the proof, take the 2-norm of \eqref{eq:line_int_graph_f}, apply Cauchy-Schwarz inequality and invoke the property that the induced two norm of block elements of a matrix such as $P_{i,-}$ is upper bounded by the induced two-norm of the full matrix $P$.
Finally, upper bound $|\persp{\alpha}{i}-\persp{\alpha}{j}|$ by $\delta$.
\end{proof}

\begin{remark}
    Although Proposition \ref{eq:prop_graph_char} is stated under the assumption \(P^{(i)} = \alpha^{(i)}P\), an analogous result can be established even without this restriction. 
\end{remark}

\begin{figure}
    \centering
    \includegraphics[width=0.9\linewidth]{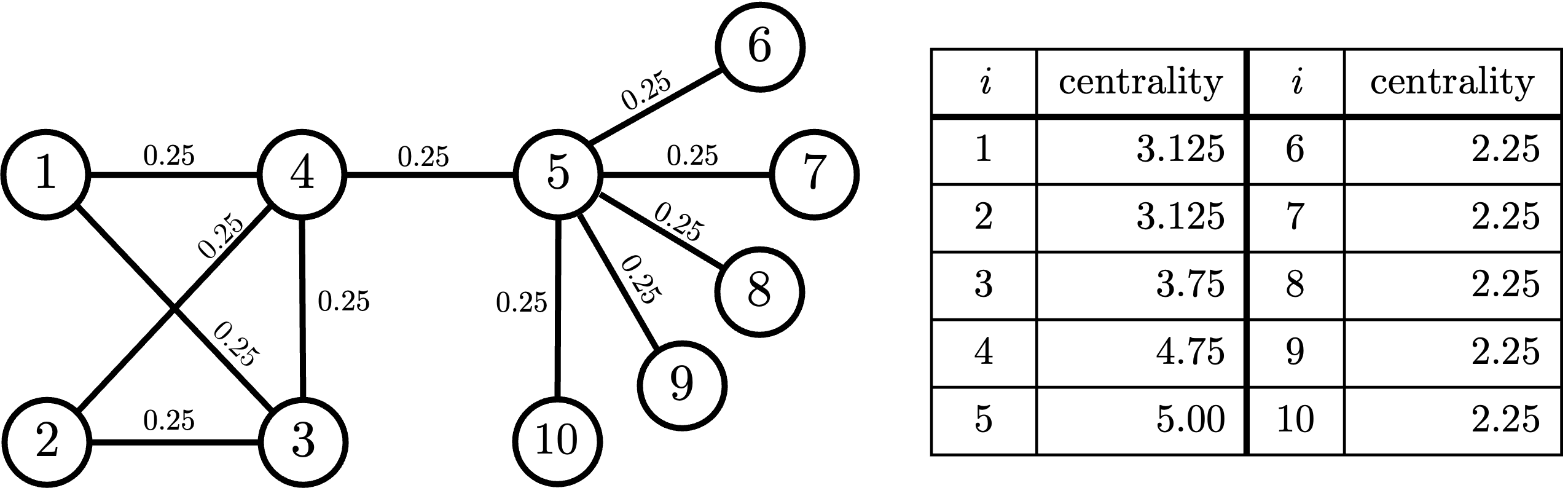}
    \caption{Example where performance gap is bigger between two players with similar network effects but smaller Bonacich {centrality} (i.e., 5 has the highest centrality, but 4 has worse loss if misaligned with 3)}
    \label{fig:network_example}
\end{figure}

\iftrue

\begin{figure*}[t!]
    \centering
    \begin{subfigure}[t]{0.32\textwidth}
        \centering
        \includegraphics[width=0.95\linewidth]{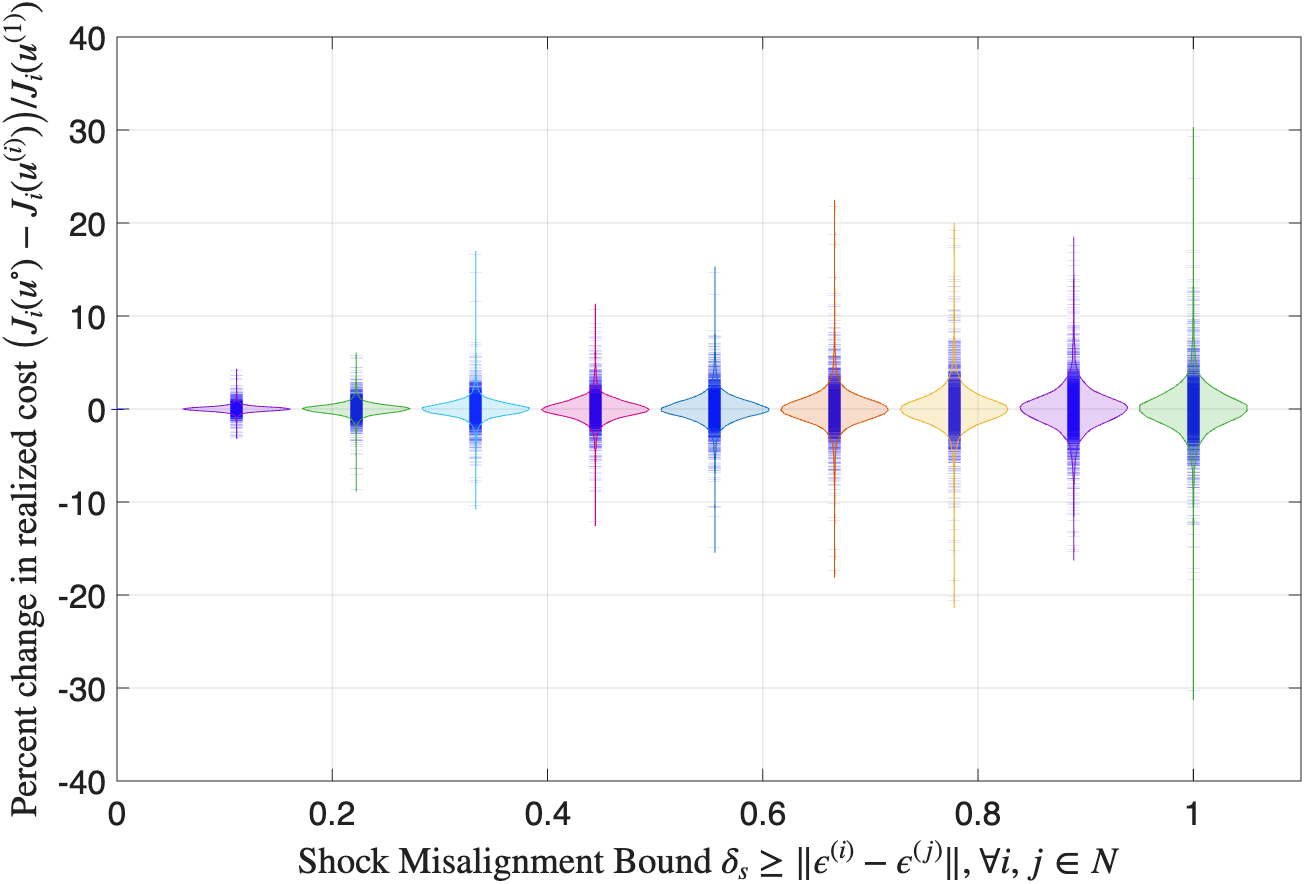}
        \caption{Relative change from predicted to realized cost with heterogeneous predicted shocks $\epsilon$.}
        \label{fig:ShockRel}
    \end{subfigure}%
    ~ 
    \begin{subfigure}[t]{0.32\textwidth}
        \centering
        \includegraphics[width=0.95\linewidth]{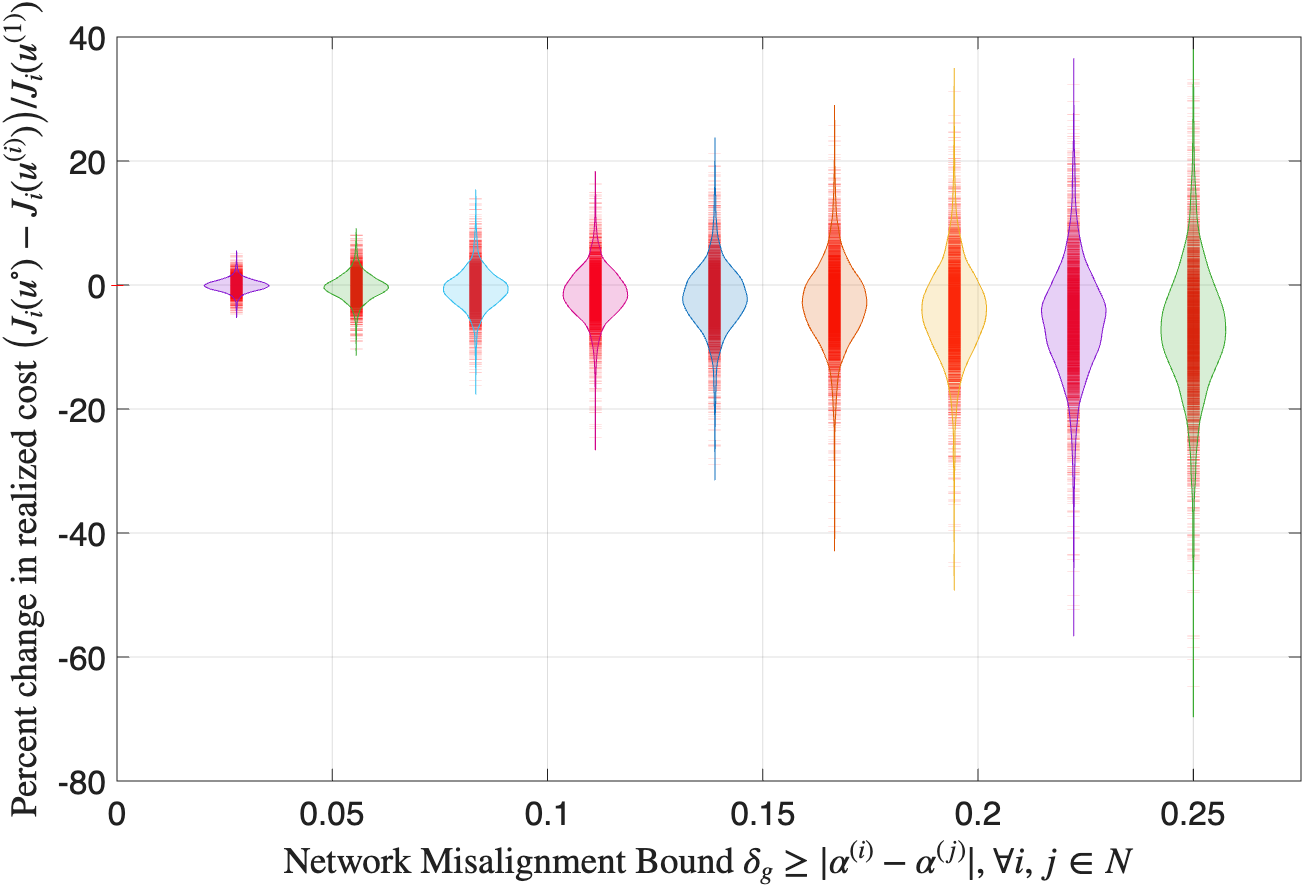}
        \caption{Relative change from predicted to realized cost with heterogeneous conjectured network $P$.}
        \label{fig:GraphRel}
    \end{subfigure}
    ~
    \begin{subfigure}[t]{0.32\textwidth}
        \centering
        \includegraphics[width=0.95\linewidth]{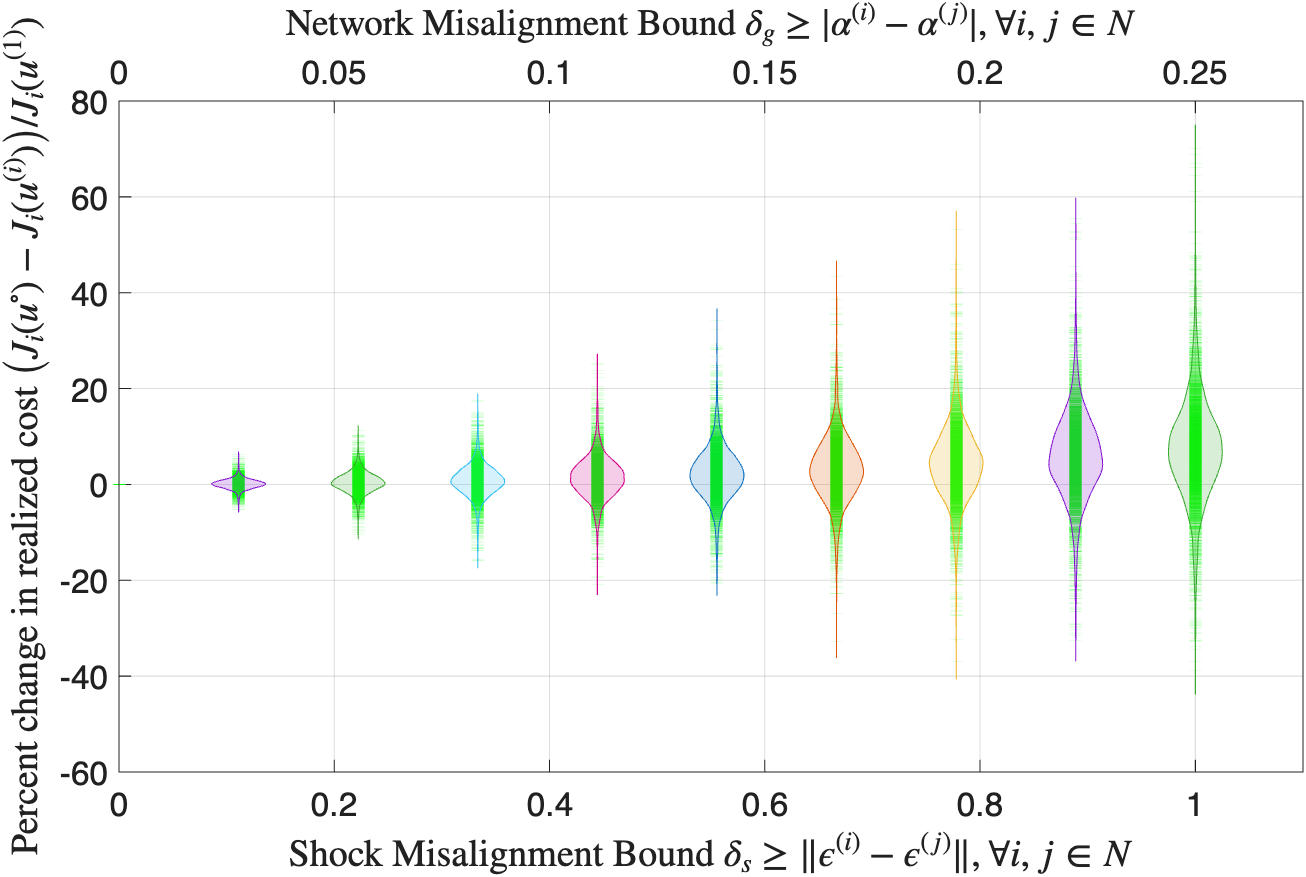}
        \caption{Relative change from predicted to realized cost with heterogeneous predicted shocks $\epsilon$ and network $P$.}
        \label{fig:BothRel}
    \end{subfigure}
    \caption{Bean plots of the Monte-Carlo simulation of network aggregative games with misconjectures. In each case, the type of possible misconjecture and thus misalignment between players varies. Along each horizontal axis is the bound on the level of respective misalignment. As the misalignment increases, the chance spread of the relative difference in predicted and realized cost grows.}
\end{figure*}

\subsubsection{Network and Shock Misalignment}
Combining the two previous cases gives the setting where players may simultaneously misconjecture about the network shock and the interaction network.
\begin{corollary}\label{cor:both}
In quadratic network aggregative games where each player $i \in N$ conjectures a network $\persp{P}{i} = \persp{\alpha}{i}P$ with $\persp{\alpha}{i} > 0$ and $\persp{\alpha}{i}(P+P^\top)/2 \prec I$, and shock to the network $\persp{\epsilon}{i} \in \mathbb{R}^m$, then the \gapname~for player $i$ is
    \begin{align*}
        &J_i(u^\circ) - J_i(\persp{u}{i}) =\\
        &\sum_{j \neq i} \persp{\epsilon}{i}\hspace{0pt}^\top L_{i,-}^\top P_{i,j}\left[\frac{\int_{x = \persp{\alpha}{i}}^{\persp{\alpha}{j}} L(x) dx}{\persp{\alpha}{j}-\persp{\alpha}{i}}\right]_{j,-}(\persp{\epsilon}{i}-\persp{\epsilon}{j})\\
        &+ \persp{\alpha}{i}(\persp{\alpha}{i}-\persp{\alpha}{j})\Big[ \frac{L(\persp{\alpha}{j})\persp{\epsilon}{j} - L(\persp{\alpha}{i})\persp{\epsilon}{i}}{\persp{\alpha}{j} - \persp{\alpha}{i}} \\
        &- \frac{\int_{x=\persp{\alpha}{i}}^{\persp{\alpha}{j}} L(x)dx}{(\persp{\alpha}{j} - \persp{\alpha}{i})^2} (\persp{\epsilon}{j}-\persp{\epsilon}{i}) \Big]_{j,-}^\top P_{i,j}^\top L(\persp{\alpha}{i})\persp{\epsilon}{i}
    \end{align*}
    where $L(x) = (I-xP)^{-1}$.
    Further, if $P$ invertible, then $\int_{x=\persp{\alpha}{i}}^{\persp{\alpha}{j}} L(x)dx = P^{-1}(L(\persp{\alpha}{i})-L(\persp{\alpha}{j})$.
\end{corollary}
The expression is garnered by evaluating the line integral over the segment of $[\overline{\epsilon}^\top,\overline{\alpha}^\top]^\top$.
Though it is difficult to draw intuitive insights from the expression in \cref{cor:both}, when each $\persp{\alpha}{i} = \persp{\alpha}{j}$ or each $\persp{\epsilon}{i} = \persp{\epsilon}{j}$, the expression reduces to the respective expressions in \cref{thm:network_agg_bound_shock} and Proposition~\ref{prop:network_agg_bound_net}.
In \cref{fig:BothRel}, we investigate the magnitude of this performance gap numerically.

\fi


\subsection{Numerical Examples}\label{subsec:num}
In this section, we seek to realize the observations and intuition generated in \cref{subsec:cons} and \cref{subsec:characterization} with real examples as well as demonstrate the prominence of this form of performance loss via Monte-Carlo Simulation.
First, we provide two examples that demonstrate the observation that the centrality of individual nodes/agents is not a determining factor on the impact of misaligned conjectures.

\subsubsection{Misaligned Shocks Example}
Consider the graph as shown in \cref{fig:network_example}.
We will consider a network game where each player has a scalar action and is associated with a single node, i.e., $u_i \in \mathbb{R}$ for all $i \in N$.
In the table in \cref{fig:network_example} is the Bonacich centrality of each node.
From the perspective of Agent 4, prior intuition would lead one to believe that Agent 5 would have the biggest impact on their performance, due to its maximal centrality.
We will illustrate two cases to show this is not true.

\noindent\textbf{Case 1:}
$\persp{\epsilon}{i} = \mathbf{1}_{10\times 1}$ for all $i \in N \setminus \{4\}$,
$\persp{\epsilon}{4} = [1,1,1,1,2,1,1,1,1,1]^\top$,
$J_4(u^\circ)-J_4(\persp{u}{4}) = 3.111$.

\noindent\textbf{Case 2:}
$\persp{\epsilon}{i} = \mathbf{1}_{10\times 1}$ for all $i \in N \setminus \{4\}$,
$\persp{\epsilon}{4} = [1,1,2,1,1,1,1,1,1,1]^\top$,
$J_4(u^\circ)-J_4(\persp{u}{4}) = 3.783$.

When juxtaposed, the above two cases demonstrate that agent 4 having an incorrect conjecture about the shock at node 3 is more impactful than having the same incorrect information about a shock at node 5 (the maximally central node).
This example realizes the observation that the \gapname~does not occur simply from individual or cumulative centrality, but rather the agents' overlap in the networks, as formalized by \cref{thm:network_agg_bound_shock}.
Such insights can prove valuable when agents must invest to acquire information to create more reliable predictions, to ensure said investment is done in the right manner.
Similar examples exist for the case where agents are misaligned on their conjectured adjacency matrices.

\subsubsection{Monte Carlo Simulation}
In \cref{fig:ShockRel} and \cref{fig:GraphRel}, we plot the results of two Monte-Carlo simulations parameterized by an upper bound on agents' conjecture misalignment of predicted shocks and network structure, respectively.

In \cref{fig:ShockRel}, we generate a network $P$ uniformly at random, remove diagonal block components, and scale to possess a maximum singular value of 0.75.
We generate each player's predicted shock uniformly at random between 0 and 1, then normalize such that the largest 2-norm difference is equal to $\delta_s$, and we then shift them by 1.
In \cref{fig:GraphRel}, we use the same network $P$ and one of the generated predicted shocks $\epsilon$, while generating a scaling factor $\persp{\alpha}{i}$ uniformly at random between $1-\delta_g/2$ and $1+\delta_g/2$ for each $i \in N$.
In \cref{fig:BothRel}, we consider both randomly generated predicted shocks bounded by $\delta_s$ and conjectured networks whose scalars are bounded by $\delta_g$.

Each plot shows the anticipated trend that larger misalignments produce larger \gapname~and that the relative gain or loss in predicted cost can be significant.
A more surprising observation is that in \cref{fig:GraphRel} as $\delta_s$, the bound on conjectured network misalignment, grows, there is a slight downward trend in the median realized cost difference; stated differently, though the spread of possible gain or loss in predicted cost increases, the likelyhood of a player doing better in the realized performance than their predicted seems to grow.
Surprisingly, the opposite is true in the case we consider both network and predicted shock misalignment in \cref{fig:BothRel}, implying that an agent may be more likely to experience greater realized cost than predicted when both forms of misalignment are present.
Though \cref{cor:both} fails to provide immediate intuition, it is clear that the case where multiple forms of conjectures are significant and worthy of greater study.

\section{Conclusion}
In this work, we introduce the \emph{game-to-real gap} as a metric to quantify the impact of model misspecification in game-theoretic settings, and analyze this metric in the context of network games. Future work will investigate on explicit characterization game-to-real gap in different class of games. Additionally, it will be interesting to investigate how an agent may seek to alter or pre-design their model to garner more robust solutions. Such an approach introduces a trade-off between nominal and robust performance. Additionally, one may apply this perspective of robustness to adaptive or data-driven decision-making by complementing these results with convergence rates or complexity bounds when conjectures are refined based on past observations.

\bibliography{refs, references_Bryce}
\bibliographystyle{IEEEtran}

\appendix
\noindent{Proof of Proposition~\ref{prop:unbounded_shock}} - 
We prove the claim by the construction of a game $G$ with three players, each with one action.
Let $\gamma \in (0,1)$, and define
$$P = \begin{bmatrix} 0 & 0 & \gamma \\ \gamma & 0 & 0 \\ 0 & \gamma & 0 \end{bmatrix},~
\persp{\epsilon}{1} = \begin{bmatrix}
    \beta\\
    1\\
    1
\end{bmatrix},~
\persp{\epsilon}{2} = \begin{bmatrix}
    1\\
    \beta\\
    1
\end{bmatrix},~
\persp{\epsilon}{3} = \begin{bmatrix}
    1\\
    1\\
    \beta
\end{bmatrix},$$
where $\beta \in (\max\{1-\delta/\sqrt{2},0\}, 1)$ is some parameter in an allowable region such that the misalignment of players predictions satisfy the supposition.
The Leontief matrix is
$$(I-P)^{-1} = \frac{1}{1-\gamma^3}\begin{bmatrix} 1 & \gamma^2 & \gamma \\ \gamma & 1 & \gamma^2 \\ \gamma^2 & \gamma & 1 \end{bmatrix},$$
from this, \eqref{eq:net_agg_equilibrium}, and the symmetry of the game, for any $i \in N$,
\begin{align*}
    J_i(u^\circ) - J_i(\persp{u}{i}) &= \persp{u}{i}_iP
    _{i,-}(\persp{u}{i} - u^\circ)\\
    &= \frac{\gamma(\gamma^2+\gamma+\beta)(1-\beta)(1-\gamma^2)}{(1-\gamma^3)^2},
\end{align*}
which is positive for $\beta \in (0,1)$ and satsifies $\lim_{\gamma \nearrow 1} \frac{\gamma(\gamma^2+\gamma+\beta)(1-\beta)(1-\gamma^2)}{(1-\gamma^3)^2} = \infty$.
\hfill\qed

\noindent\emph{Proof of Proposition~\ref{prop:unbounded_graph}} - 
Define 
\begin{align*}
P^{(1)} = \begin{bmatrix} 0 & 0 & 1-\frac{\delta}{\sqrt{2}} \\ 1  & 0 & 0 \\ 0 & 1  & 0 \end{bmatrix}&, P^{(2)} = \begin{bmatrix} 0 & 0 & 1 \\ 1-\frac{\delta}{\sqrt{2}}  & 0 & 0 \\ 0 & 1  & 0 \end{bmatrix}, \\ 
P^{(3)} &= \begin{bmatrix} 0 & 0 & 1 \\ 1  & 0 & 0 \\ 0 & 1-\frac{\delta}{\sqrt{2}}  & 0 .\end{bmatrix}
\end{align*}
Note that \(\|P^{(i)}-P^{(j)}\|_F = \delta.\)

Note that 
\begin{align*}
(I-P^{(1)})^{-1} = \frac{\sqrt{2}}{\delta}\begin{bmatrix}
        1 & 1-\frac{\delta}{\sqrt{2}}  & 1-\frac{\delta}{\sqrt{2}}  \\ 
         1  &  1 & 1-\frac{\delta}{\sqrt{2}}  \\ 
         1 & 1  & 1 
    \end{bmatrix},
\\ 
    (I-P^{(2)})^{-1} = \frac{\sqrt{2}}{\delta}\begin{bmatrix}
        1 & 1  & 1 \\ 
         1-\frac{\delta}{\sqrt{2}}  &  1 & 1-\frac{\delta}{\sqrt{2}} \\ 
         1-\frac{\delta}{\sqrt{2}} & 1  & 1 
    \end{bmatrix}, \\
    (I-P^{(3)})^{-1} = \frac{\sqrt{2}}{\delta}\begin{bmatrix}
        1 & 1-\frac{\delta}{\sqrt{2}}  & 1 \\ 
         1  &  1 & 1 \\ 
         1-\frac{\delta}{\sqrt{2}} & 1-\frac{\delta}{\sqrt{2}}  & 1
    \end{bmatrix}. 
\end{align*}


Define \(\epsilon = [\gamma,-\gamma, -\gamma]^{\top}\). Consequently, 
\begin{align*}
    u^{(1)} = \gamma \frac{\sqrt{2}}{\delta}\begin{bmatrix}
        \sqrt{2\delta}-1 \\ \frac{\delta}{\sqrt{2}}-1\\-1
    \end{bmatrix}, u^{(2)} = \gamma \frac{\sqrt{2}}{\delta}\begin{bmatrix}
       -1 \\ -1\\-1 -\frac{\delta}{\sqrt{2}}
    \end{bmatrix}, \\ 
     u^{(3)} = \gamma \frac{\sqrt{2}}{\delta}\begin{bmatrix}
       -1+  \frac{\delta}{\sqrt{2}} \\-1 \\ -1
    \end{bmatrix},  u^{\circ} = \gamma \frac{\sqrt{2}}{\delta}\begin{bmatrix}
        \sqrt{2\delta}-1 \\ -1\\-1
    \end{bmatrix}.
\end{align*}

Next, we compute for \(i =2.\)
\begin{align*}
    &J_2(u^\circ) - J_2(u^{(2)}) = u_2^{(2)}\left(\left(P^{(2)}\persp{u}{2}\right)_2 - \left(P^{(2)}u^\circ\right)_2\right)
    \\
    &=2\gamma^2(\sqrt{2}-\delta)/\delta
\end{align*}
Selecting $\delta \in (0,\sqrt{2})$ makes the gap positive and quadratic in $\gamma$.
To conclude, as \(\gamma\) increases the gap will increase. 


\hfill\qed

\end{document}